%% file: main.tex
\documentclass[USenglish,a4paper,11pt]{scrartcl}

\usepackage[utf8]{inputenc}
\usepackage{lmodern}

\usepackage{amsmath}
\usepackage{amsthm}
\usepackage{thmtools}
\usepackage{amsfonts}
\usepackage{amssymb}
\usepackage{url}

\usepackage{forest}
\usepackage{relsize}
\usetikzlibrary{trees, snakes}
\usetikzlibrary{matrix}

\theoremstyle{plain}
\newtheorem*{theorem*}{Theorem}
\newtheorem*{lemma*}{Lemma}
\newtheorem{theorem}{Theorem}

\newtheorem{lemma}{Lemma}

\newcommand{\norone}[1]{\left\|#1\right\|_1}
\newcommand{\norinf}[1]{\left\|#1\right\|_{\infty}}
\newcommand{\ZZ}{\mathbb{Z}}
\newcommand{\ZZgeq}{\mathbb{Z}_{\geq 0}}
\newcommand{\RR}{\mathbb{R}}
\newcommand{\RRgeq}{\mathbb{R}_{\geq 0}}
\newcommand{\BB}{B}

\title{About the Complexity of Two-Stage Stochastic IPs \thanks{This work was mostly done during the authors time at EPFL. The project was supported by the Swiss National Science Foundation (SNSF) within the project Convexity, geometry of numbers, and the complexity of integer programming (Nr.163071)}}
\author{Kim-Manuel Klein\\
University of Kiel\\
\url{kmk@informatik.uni-kiel.de}}

\begin{document}

\date{}
\maketitle

\begin{abstract}
We consider so called $2$-stage stochastic integer programs (IPs) and their generalized form of multi-stage stochastic IPs. A $2$-stage stochastic IP is an integer program of the form $\max \{ c^T x \mid Ax = b, l \leq x \leq u, x \in \ZZ^{nt + s} \}$ where the constraint matrix $A \in \ZZ^{r \times s}$ consists roughly of $n$ repetition of a block matrix $A$ on the vertical line and $n$ repetitions of a matrix $B \in \ZZ^{r \times t}$ on the diagonal.
   
In this paper we improve upon an algorithmic result by Hemmecke and Schultz form 2003 \cite{Hemmecke_two_stage_03} to solve $2$-stage stochastic IPs. The algorithm is based on the Graver augmentation framework where our main contribution is to give an explicit doubly exponential bound on the size of the augmenting steps. The previous bound for the size of the augmenting steps relied on non-constructive finiteness arguments from commutative algebra and therefore only an implicit bound was known that depends on parameters $r,s,t$ and $\Delta$, where $\Delta$ is the largest entry of the constraint matrix.
Our new improved bound however is obtained by a novel theorem which argues about the intersection of paths in a vector space.
   
As a result of our new bound we obtain an algorithm to solve $2$-stage stochastic IPs in time $poly(n,t) \cdot f(r,s,\Delta)$, where $f$ is a doubly exponential function.

    
To complement our result, we also prove a doubly exponential lower bound for the size of the augmenting steps.
\end{abstract}

\section{Introduction}
Integer programming is one of the most fundamental problems in algorithm theory. Many problems in combinatorial optimization and other areas can be modeled as an integer program . An \emph{integer program} (IP) is thereby of the form
\begin{align*}
    \max \{ c^T x \mid Ax = b, l \leq x \leq u, x \in \ZZ^n \}
\end{align*}
for some matrix $A \in \ZZ^{m \times n}$, a right hand side $b \in \ZZ^m$, a cost vector $c \in \ZZ^n$ and lower and upper bounds $l, u \in \ZZ^n$. 
The famous algorithm of Kannan \cite{kannan1987minkowski} computes an optimal solution of the IP in time of roughly $n^{O(n)} \cdot poly(m,\log \Delta)$, where $\Delta$ is the largest entry of $A$ and $b$. 

In recent years there was significant progress in the development of algorithms for IPs when the constraint matrix $A$ has a specific structure. Consider for example the class of integer programs with a constraint matrix $\mathcal{N}$ of the form
\begin{align*}
    \mathcal{N} = \begin{pmatrix}
A & A & \cdots & A \\
B & 0 & \cdots & 0\\
0 & B & \ddots &  \vdots \\
\vdots & \ddots & \ddots &  0 \\
0 & \cdots & 0 & B
\end{pmatrix} 
\end{align*}
for some block matrices $A \in \ZZ^{r \times s}$ and $B \in \ZZ^{r \times t}$. An IP of this specific structure is called an \emph{$n$-fold} IP. This class of IPs has found numerous applications in the area of string algorithms \cite{knop2017combinatorial}, social choice games \cite{koutecky2018IJCAI, KnopKM17-bribery} and scheduling \cite{jansen2018scheduling, knop2017scheduling}.
State-of-the-art algorithms compute a solution of an $n$-fold IP in time $poly(n,t) \Delta^{O(r^2s)}$ \cite{eisenbrand18, jansen2018near_linear, KouteckyLO18}, where $\Delta$ is the largest entry in matrices $A$ and $B$.

\subsection{Two-Stage Stochastic Integer Programming}
Stochastic programming deals with uncertainty of decision making over time \cite{kall1994stochastic}. One of the basic models in stochastic programming is $2$-stage stochastic programming. In this model one has to decide on a solution at the first stage and in the second stage there is an uncertainty where $n$ possible scenarios can happen. Each of $n$ possible scenarios might have a different optimal solution and the goal is to minimize the costs of the solution of the first stage in addition to the expected costs of the solution of the second stage. In the case that said scenarios can be modeled by an (integer) linear program, we are talking about \emph{$2$-stage stochastic (integer) linear programs}. $2$-stage stochastic linear programs that do not contain any integer variable are well understood (we refer to standard text books~\cite{birge2011introduction, kall1994stochastic}). In contrast, $2$-stage stochastic programs that contain integer variables are hard to solve and are topic of ongoing research. Typically, those IPs are investigated in the context of decomposition based methods (we refer to a tutorial \cite{kuccukyavuz2017introduction} or a survey \cite{twostage_survey} on the topic). For recent progress on $2$-stage stochastic programs we refer to \cite{twostage_branch_and_bound, caroe1998shaped, twostage_survey}.
The interest in solving $2$-stage stochastic (I)LPs efficiently stems from their wide range of applications for example in modeling manufacturing processes \cite{dempster1981analytical} or energy planing \cite{haneveld2001optimizing}.

In this paper we consider $2$-stage stochastic IPs with only integral variables. This so called pure integral $2$-stage stochastic IPs have also been considered in the literature from a practical perspective (see \cite{gade2014decomposition,zhang2014finitely}).
The considered IP is then of the form
\begin{align} \label{IP:2-stage}
    \max c^T x \\ \notag
    \mathcal{A}x = b\\ \notag
    l \leq x \leq u\\ \notag
    x \in \ZZ^{s + nt} \notag
\end{align}
for given objective vector $c \in \ZZgeq^{s + nt}$ upper and lower bound $\ell, u \in \ZZgeq^{s + nt}$.
The constraint matrix $\mathcal{A}$ has the shape
\begin{align*}
    \mathcal{A} = \begin{pmatrix}
A^{(1)} & B^{(1)} & 0 & \cdots & 0\\
A^{(2)} & 0 & B^{(2)} & \ddots &  \vdots \\
\vdots & \vdots & \ddots & \ddots &  0 \\
A^{(n)} & 0 & \cdots & 0 & B^{(n)}
\end{pmatrix} 
\end{align*}
for given block matrices $A^{(1)}, \ldots , A^{(n)} \in \ZZ^{r \times s}$ and $B^{(1)}, \ldots , B^{(n)} \in \ZZ^{r \times t}$.
Typically, $2$-stage stochastic IPs are written in a slightly different (equivalent) form that explicitly involves the scenarios and the probability distribution of the scenarios of the second stage.
In this presented form, roughly speaking, the solution for the first stage scenario is encoded in the variables corresponding to vertical block matrices. A solution for each of the second stage scenarios is encoded in the variables corresponding to one of the diagonal block matrices and the expectation for the second stage scenarios can be encoded in a linear objective function. 
Since we do not rely on known techniques of stochastic programming in this paper, we omit the technicalities surrounding $2$-stage stochastic IPs and simply refer to a survey for further details \cite{twostage_survey}.

Despite their similarity, it seems that $2$-stage IPs are significantly harder to solve than $n$-fold IPs. While it is known that the $2$-stage stochastic IP with constraint matrix $\mathcal{S}$ can be solved in running time of the form $poly(n) \cdot f(r,s,t, \Delta)$ for some computable function $f$, which was developed by Hemmecke and Schultz \cite{Hemmecke_two_stage_03}, the actual dependency on the parameters $r,s,t, \Delta$ was unknown (we elaborate on this further in the coming section). Their algorithm is based on the augmentation framework which we also discuss in the following section.

\subsection{The Augmentation Framework}
Suppose we have an initial feasible solution $z_0$ of an IP $\max \{ c^T x \mid Ax = b, l \leq x \leq u, x \in \ZZ^n \}$ and our goal is to find an optimal solution. The idea behind the augmentation framework (see \cite{loera2012algebraic}) is to compute an augmenting (integral) vector $y$ in the kernel, i.e. $y \in ker(A)$ with $c^T y > 0$. A new solution $z'$ with improved objective can then be defined by $z' = z_0 + \lambda y$ for a suitable $\lambda \in \ZZgeq$. This procedure can be iterated until a solution with optimal objective is obtained eventually.

We call an integer vector $y \in ker(A)$ a \emph{cycle}. A cycle can be decomposed if there exist integral vectors $u,v \in ker(A)\setminus \{ 0 \}$ with $y = u + v$ and $u_i \cdot v_i \geq 0$ for all $i$ (i.e. the vectors are sign-compatible with $y$). An integral vector $y \in ker(A)$ that can not be decomposed is called a \emph{Graver element} \cite{Graver1975foundations} or we simply say that it is \emph{indecomposable}. The set of all indecomposable elements is called the \emph{Graver basis}.

The power of the augmentation framework is based on the observation that the size of Graver elements and therefore also the size of the Graver basis can be bounded. With the help of these bounds, good augmenting steps can be computed by a dynamic program and finally the corresponding IP can be solved efficiently.

In the case that the constraint matrix has a very specific structure, one can sometimes show improved bounds. Specifically, if the constraint matrix $A$ has a $2$-stage stochastic shape with identical block matrices in the vertical and diagonal line, then Hemmecke and Schultz~\cite{Hemmecke_two_stage_03} were able to prove a bound for the size of Graver elements that only depends on the parameters $r,s,t$ and $\Delta$. The presented bound is an existential result and uses so called saturation results from commutative algebra. As MacLagan's theorem is used in the proof of the bound no explicit function can be derived. It is only known that the dependency on the parameters is lower bounded by ackerman's function \cite{ackermanian2009}. This implies that the implicit bound for the size of Graver elements by Hemmecke and Schultz can not be improved beyond an ackermanian dependency in the parameters $r,s,t$ and $\Delta$.

In a very recent paper it was even conjectured that an algorithm with an explicit bound on parameters $r,s,t$ and $\Delta$ in the running time to solve IPs of the form (\ref{IP:2-stage}) does not exist \cite{knop2018tight}.

Very recently, improved bounds for Graver elements of general matrices and matrices with specific structure like $n$-fold \cite{eisenbrand18} or $4$-block structure \cite{chen2018Graver} were developed. They are based on the Steinitz Lemma, which was previously also used by Eisenbrand and Weismantel \cite{eisenbrand2018proximity} in the context of integer programming.
\begin{lemma}[Steinitz \cite{grinberg1980value, steinitz1913bedingt}]
    Let $v_1, \ldots , v_n \in \RR^m$ be vectors with $\norinf{v_i} \leq \Delta$ for $1 \leq i \leq n$. Assuming that $\sum_{i=1}^n v_i = 0$ then there is a permutation $\Pi$ such that for each $k \in \{1, \ldots , n \}$ the norm of the partial sum $\norinf{\sum_{i=1}^k v_{\Pi(i)}}$ is bounded by  $m \Delta$
\end{lemma}
The Steinitz Lemma was used by Eisenbrand, Hunkenschröder and Klein \cite{eisenbrand18} to bound the size of Graver elements for a given matrix $A$. As we use the following theorem and its technique in this paper, we give a brief sketch of its proof.
\begin{theorem}[Eisenbrand, Hunkenschröder, Klein \cite{eisenbrand18}] \label{thm:Graver_bound}
    Let $A \in \ZZ^{m \times n}$ be an integer matrix where every entry of $A$ is bounded by $\Delta$ in absolute value. Let $g \in \ZZ^n$ be an element of the Graver Basis of $A$ then $\norone{g} \leq (2m\Delta +1)^m$.
\end{theorem}
\begin{proof}
    Consider the sequence of vectors $v_1, \ldots , v_{\norone{g}}$ consisting of $y_i$ copies of the $i$-th column of $A$ if $g_i$ is positive and $|g_i|$ copies of the negative of the $i$-th coplumn of $A$ if $g_i$ is negative. As $g$ is a Graver element we obtain that $v_1 + \ldots + v_{\norone{g}} = 0$.
    Using the Steinitz Lemma above, there exists a reordering $u_1 + \ldots + u_{\norone{g}}$ of the vectors such that the partial sum $p_k = \norinf{\sum_{i=1}^{k} u_i} \leq \Delta m$ for each $k \leq \norone{g}$.
    
    Suppose by contradiction that $\norone{g} > (2m\Delta +1)^m$. Then by the pigeonhole principle there exist two partial sums that sum up to the same value. However, this means that $g$ can be decomposed and hence can not be a Graver element.
\end{proof}

\subsection{Our Results:}
The main result of this paper is to prove a new structural lemma that enhances the toolset of the augmentation framework. We show that this Lemma can be directly used to obtain an explicit bound for Graver elements of the constraint matrix of $2$-stage stochastic IPs. But we think that it might also be of independent interest as it provides interesting structural insights in vector sets.
\begin{lemma}\label{lem:subrepresentation}
    Given multisets $T_1, \ldots , T_n \subset \ZZgeq^d$ where all elements $t \in T_i$ have bounded size $\norinf{t} \leq \Delta$. Assuming that the total sum of all elements in each set is equal, i.e.
    \begin{align*}
        \sum_{t \in T_1} t = \ldots = \sum_{t \in T_n} t
    \end{align*}
    then there exist nonempty submultisets $S_1 \subseteq T_1, \ldots , S_n \subseteq T_n$ of bounded size $|S_i| \leq (d \Delta)^{O(d (\Delta^{d^2}))}$ such that
    \begin{align*}
        \sum_{s \in S_1} s = \ldots = \sum_{s \in S_n} s.
    \end{align*}
\end{lemma}
Note that this lemma only makes sense when we consider the $T_i$ to be multisets as the number of different sets without allowing multiplicity of vectors would be bounded by $2^{\Delta^d}$. 

A geometric interpretation of the lemma is given in the following figure. On the left side we have $n$-paths consisting of sets of vectors and all path end at the same point $b$.
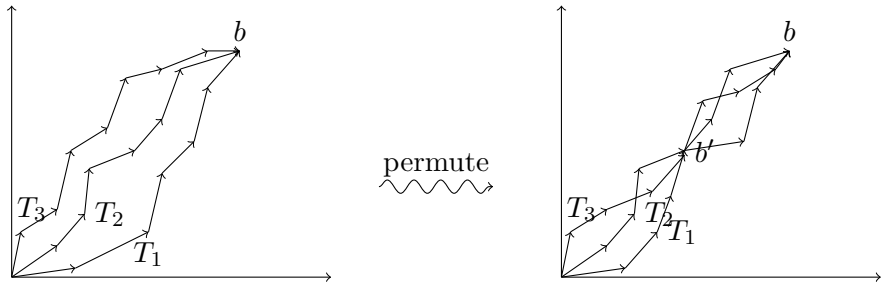
\begin{center}
    \begin{tabular}{ccc}
    \input{fig_paths.tex} & \input{fig_pathsperm.tex}\\
\end{tabular}
\end{center}
Then the Lemma shows, that there always exist permutations of the vectors of each path such that all paths meet at a point $b'$ of bounded size. The bound does only depend on $\Delta$ and the dimension $d$ and is thus independent of the number of paths $n$ and the size of $b$.
For the proof of the Lemma we need basic properties for the intersection of integer cones. We show that those properties can be obtained by using the Steinitz Lemma.

We show that Lemma \ref{lem:subrepresentation} has strong implications in the context of integer programming. Using the Lemma, we can show that the size of Graver elements of matrix $\mathcal{A}$ is bounded by $(r s \Delta)^{O(r s ((2 r \Delta +1)^{r s^2}))}$. Within the framework of Graver augmenting steps the bound implies that $2$-stage stochastic IPs can be solved in time $n^2t^2 \varphi \log^2 (nt)(r s \Delta)^{O(r s^2 ((2 r \Delta +1)^{r s^2}))}$, where $\varphi$ is the encoding length of the instance.
With this we improve upon an implicit bound for the size of the Graver elements matrix $2$-stage stochastic constraint matrices due to Hemmecke and Schultz \cite{Hemmecke_two_stage_03}. 

Furthermore, we show that our Lemma can also be applied to bound the size of Graver elements of constraint matrices that have a multi-stage stochastic structure. Multi-stage stochastic IPs are a well known generalization of $2$-stage stochastic IPs. By this, we improve upon a result of Aschenbrenner and Hemmecke \cite{hemmecke_multi-stage}.

To complement our results for the upper bound, we also present a lower bound for the size of Graver elements of matrices that have a $2$-stage stochastic IP structure. The given lower bound is for the case of $r=1$. In this case we present a matrix where the Graver elements have a size of $2^{\Omega(\Delta^s)}$.

\section{The Complexity of Two-Stage Stochastic IPs}\label{sec:two-stage}

First, we argue about the application of our main Lemma \ref{lem:subrepresentation}.
In the following we show that the infinity-norm of Graver elements of matrices with a $2$-stage stochastic structure can be bounded by using the lemma.

Given the block structure of the IP \ref{IP:2-stage}, we define for a vector $y \in \ZZ^{s + nt}$ with $\mathcal{A}y = 0$ the vector $y^{(0)} \in \ZZgeq^{s}$ which consists of the entries of $y$ that belong to the vertical block matrices $A^{(i)}$ and we define $y^{(i)} \in \ZZgeq^t$ to be the entries of $y$ that belong to the diagonal block matrix $B^{(i)}$.
\begin{theorem} \label{thm:2SSIP_bound}
    Let $y$ be a Graver element of the constraint matrix $\mathcal{A}$ of IP (\ref{IP:2-stage}). Then $\norinf{y}$ is bounded by $(r s \Delta)^{O(r s ((2 r \Delta +1)^{r s^2}))}$.
    More precisely, $\norone{y^{(i)}} \leq (r s \Delta)^{O(r s ((2 r \Delta +1)^{r s^2}))}$ for every $0 \leq i \leq n$.
\end{theorem}
\begin{proof}
    Let $y \in \ZZgeq^{s+nt}$ be a cycle of IP (\ref{IP:2-stage}), i.e. $\mathcal{A} y = 0$.
    Consider a submatrix of the matrix $\mathcal{A}$ denoted by $(A^{(i)} B^{(i)}) \in \ZZ^{r \times (s+t)}$ consisting of the block matrix $A^{(i)}$ of the vertical line and the block matrix $B^{(i)}$ of the diagonal line. Consider further the corresponding variables $v^{(i)} = \begin{pmatrix}
    y^{(0)}\\ y^{(i)} \end{pmatrix} \in \ZZ^{s+t}$ of the respective matrix $A^{(i)}$ and $B^{(i)}$. Since $\mathcal{A}y = 0$, we also have that $( A^{(i)} B^{(i)} ) v^{(i)} = 0$. Hence, we can decompose $v^{(i)}$ into a multiset $C_i$ of indecomposable elements, i.e. $v^{(i)} = \sum_{c\in C_i} c$. By Lemma \ref{thm:Graver_bound} we obtain the bound  $\norone{c} \leq (2r \Delta +1)^{r}$ for each $c \in C_i$.
    
    Since all matrices $(A^{(i)} B^{(i)})$ share the same set of variables in the overlapping block matrices $A^{(i)}$, we can not choose indecomposable elements independently in each block to obtain a cycle of smaller size for the entire matrix $\mathcal{A}$. 
    Let $p: \ZZ^{s+t} \to \ZZ^s$ be the projection that maps a cycle $c$ of a block matrix $(A^{(i)} B^{(i)})$ to the variables in the overlapping part, i.e. $p(c) = p(\begin{pmatrix} c^{(0)} \\ c^{(i)} \end{pmatrix}) = c^{(0)}$. 
    In the case that $\norinf{y}$ is large we will show that we can find a cycle $\bar{y}$ of smaller length and $\bar{y} \leq y$. In order to obtain this cycle $\bar{y}$ for the entire matrix $\mathcal{A}$, we have to find a multiset of cycles $\bar{C}_i \subset C_i$ in each block matrix $(A^{(i)} B^{(i)})$ such that the sum of the projected parts is identical, i.e. $\sum_{c \in \bar{C}_1} p(c) = \ldots = \sum_{c \in \bar{C}_n} p(c)$. 
    We apply Lemma \ref{lem:subrepresentation} to the multisets $p(C_1), \ldots , p(C_n)$, where $p(C_i) = \{p(c) \mid c \in C_i \}$ is the multiset of projected elements in $C_i$ with $\norone{p(c)} \leq (2r \Delta +1)^{r}$. Note that $\sum_{x \in p(C_1)} x = \ldots = \sum_{x \in p(C_n)} x = y^{(0)}$ and hence the conditions to apply Lemma \ref{lem:subrepresentation} are fulfilled. Since every $v^{(i)}$ is decomposed in a sign compatible way, every entry of the vector in $p(C_i)$ has the same sign. Hence we can flip the negative signs in order to apply Lemma \ref{lem:subrepresentation}. 
    
    By the Lemma, there exist submultisets $S_1 \subseteq p(C_1), \ldots , S_n \subseteq p(C_n)$ such that $\sum_{x \in S_1} x = \ldots = \sum_{x \in S_n} x$ and $|S_i| \leq (s \norone{c})^{O(s (\norone{c}^{s^2}))} = (r s \Delta)^{O(r s ((2 r \Delta +1)^{r s^2}))}$. As there exist submultisets $\bar{C}_1 \subseteq C_1, \ldots \bar{C}_n \subseteq C_n$ with $p(\bar{C_1}) = S_1, \ldots p(\bar{C}_n) = S_n$, we can use those submultisets $\bar{C}_i$ to define a solution $\bar{y} \leq y$. For $i>0$ let $\bar{y}^{(i)} = \sum_{c \in \bar{C}_i} \bar{p}(c)$, where $\bar{p}(c)$ is the projection that maps a cycle $c \in \bar{C}_i$ to the part that belongs to matrix $B^{(i)}$, i.e. $\bar{p}(\begin{pmatrix} c^{(0)} \\ c^{(i)} \end{pmatrix}) = c^{(i)}$. And let $\bar{y}^{(0)} = \sum_{c \in \bar{C}_i} p(c)$ for an arbitrary $i>0$, which is well defined as the sum is identical for all multisets $\bar{C}_i$.
    As the cardinality of the multisets $\bar{C}_i$ is bounded, we know by construction of $\bar{y}$ that the one-norm of every $y^{(i)}$ is bounded by 
    \begin{align*}
        \norone{y^{(i)}} \leq (2r \Delta +1)^{r} \cdot (r s \Delta)^{O(r s ((2 r \Delta +1)^{r s^2}))} = (r s \Delta)^{O(r s ((2 r \Delta +1)^{r s^2}))}.
    \end{align*}
    This directly implies the infinity-norm bound for $y$ as well.
\end{proof}

\subsubsection*{Computing the Augmenting Step}
As a direct consequence of the bound for the size of the Graver elements, we obtain by the framework of augmenting steps an efficient algorithm to compute an optimal solution of a $2$-stage stochastic IP. 
For this we can use the algorithm by Hemmecke and Schultz \cite{Hemmecke_two_stage_03} or a more recent result by Koutecky, Levin and Onn \cite{KouteckyLO18} which gives a strongly polynomial algorithm. Using these algorithms directly would result in an algorithm with a running time of the form $poly(n)\cdot f(r,s,t, \Delta)$ for some doubly exponential function involving parameters $r,s,t$ and $\Delta$. 
However, in the following we explain briefly how the augmenting step can be computed in order to obtain an algorithm with a running time that is polynomial in $t$.

Given a feasible solution $z \in \ZZgeq^{s + nt}$ of IP (\ref{IP:2-stage}) and a multiple $\lambda \in \ZZgeq$ (which can be guessed). A core ingredient in the augmenting framework is to find an augmenting step. Therefore, we have to compute a Graver element $y \in ker(\mathcal{A})$ such that $z+ \lambda y$ is a feasible solution of IP (\ref{IP:2-stage}) and the objective $\lambda c^T y$ is maximal over all Graver elements.

Let $L = (r s \Delta)^{O(r s ((2 r \Delta +1)^{r s^2}))}$ be the bound for $\norone{y^{(i)}}$ that we obtain from the previous Lemma. To find the optimal augmenting step, it is sufficient to solve the IP $\max \{c^T x \mid \mathcal{A}x = 0, \bar{\ell} \leq x \leq \bar{u} , \norinf{x} \leq L \}$ for modified upper and lower bounds $\bar{\ell}, \bar{u}$ according to the multiple $\lambda$ and the feasible solution $z$.
Having the best augmenting step at hand, one can show that the objective value improves by a factor of $1-\frac{1}{2n}$. This is due to the fact that the difference $z-z^*$ between $z$ and an optimal solution $z^*$ can be represented by
\begin{align*}
    z -z^* = \sum_{i=1}^{2n} \lambda_i g_i
\end{align*}
for Graver elements $g_1, \ldots g_{2n} \in \ZZgeq^d$ and multiplicities $\lambda_1, \ldots , \lambda_{2n} \in \ZZgeq$ \cite{cook1986integer}.

In the following we briefly show how to solve the IP $\max \{c^T x \mid \mathcal{A}x = 0, \bar{\ell} \leq x \leq \bar{u} , \norinf{x} \leq L \}$ in order to compute the augmenting step. The algorithm works as follows:
\begin{itemize}
    \item Compute for every $y^{(0)}$ with $\norone{y^{(0)}} \leq L$ the objective value of the cycle $y$ consisting of $y^{(0)}, \bar{y}^{(1)}, \ldots , \bar{y}^{(n)}$, where $\bar{y}^{(i)}$ for $i>0$ are the optimal solutions of the IP
    \begin{align*}
        \max & (c^{(i)})^T \bar{y}^{(i)} \\
        B^{(i)}\bar{y}^{(i)} & = - A^{(i)} y^{(0)}\\
        \bar{\ell}^{(i)} \leq & \bar{y}^{(i)} \leq \bar{u}^{(i)}
    \end{align*}
    where $\bar{\ell}^{(i)}, \bar{u}^{(i)}$ are the upper and lower bounds for the variables $\bar{y}^{(i)}$ and $c^{(i)}$ their corresponding objective vector. Note that the first set of constraints of the IP ensure that $\mathcal{A}y = 0$. The IPs can be solved with the algorithm of Eisenbrand and Weismantel \cite{eisenbrand2018proximity} in time $O(\Delta^{O(r^2)})$ each.
    
    \item Return the cycle with maximum objective.
\end{itemize}
As the number of different vectors $y^{(0)}$ with $1$-norm $\leq L$ is bounded by $(L+1)^s = (r s \Delta)^{O(r s^2 ((2 r \Delta +1)^{r s^2}))}$ step 1 of the algorithm takes time $\Delta^{O(r^2)} \cdot (r s \Delta)^{O(r s^2 ((2 r \Delta +1)^{r s^2}))}$. 

Putting all things together, we obtain the following theorem regarding the worst-case complexity for solving $2$-stage stochastic IPs. For details regarding the remaining parts of the augmenting framework like finding an initial feasible solution or a bound on the required augmenting steps we refer to \cite{eisenbrand18} and \cite{KouteckyLO18}
\begin{theorem}
    A $2$-stage stochastic IP of the form (\ref{IP:2-stage}) can be solved in time
    \begin{align*}
        n^2t^2 \varphi \log^2 (nt)(r s \Delta)^{O(r s^2 ((2 r \Delta +1)^{r s^2}))},
    \end{align*} where $\varphi$ is the encoding length of the IP.
\end{theorem}

\subsection{About the Intersection of Integer Cones}
Before we are ready to prove our main Lemma \ref{lem:subrepresentation}, we need two helpful observations about the intersection of integer cones. An integer cone is defined for a given (finite) generating set $B \subset \ZZgeq^d$ of elements by
\begin{align*}
    int.cone(B) = \{ \sum_{b \in B} \lambda_b b \mid \lambda \in \ZZgeq^B \}.    
\end{align*}
Note that the intersection of two integer cones is again an integer cone, as the intersection is closed under addition and scalar multiplication of positive integers.

We say that an element $b$ of an integer cone $int.cone(B)$ is \emph{indecomposable} if there do not exist elements $b_1, b_2 \in int.cone(B) \setminus \{ 0\}$ such that $b = b_1 + b_2$.
We can assume that the generating set $B$ of an integer cone consists just of the set of indecomposable elements as any decomposable element can be removed from the generating set.

In the following we allow to use a vector set $B$ as a matrix and vice versa where the elements of the set $B$ are the columns of the matrix $B$. This way we can multiply $B$  with a vector, i.e. $B \lambda =  \sum_{b \in B} \lambda_b b$ for some $\lambda \in \ZZ^B$.
\begin{lemma}\label{lem:intersectiontwo}
    Given two integer cones $int.cone(B^{(1)})$ and $int.cone(B^{(2)})$ for some generating sets $B^{(1)}, B^{(2)} \subset \ZZ^{d}$ where each element $x \in B^{(1)} \cup B^{(2)}$ has bounded norm $\norinf{x} \leq \Delta$. Consider the integer cone of the intersection 
    \begin{align*}
        int.cone(\hat{B}) = int.cone(B^{(1)}) \cap int.cone(B^{(2)})
    \end{align*}
    for some generating set of elements $\hat{B}$.
    Then for each generating element $b \in \hat{B}$ of the  intersection cone with $b = B^{(1)} \lambda = B^{(2)} \gamma$ for some $\lambda \in \ZZgeq^{B^{(1)}}$ and $\gamma \in \ZZgeq^{B^{(2)}}$, we have that $\norone{\lambda}, \norone{\gamma} \leq (2d \Delta +1)^d$. Furthermore, the size of $b$ is bounded by $\norinf{b} \leq \Delta (2d \Delta +1)^d$
\end{lemma}
\begin{proof}
Consider the representation of a point $b = B^{(1)} \lambda = B^{(2)} \gamma$ in the intersection of $int.cone(B^{(1)})$ and $int.cone(B^{(2)})$. The sum $v_1 + \ldots v_{(\norone{\lambda} + \norone{\gamma})}$ consisting of 
$\lambda_i$ copies of the $i$-th element of $B^{(1)}$ and $\gamma_i$ copies of the negative of the $i$-th element of $B^{(2)}$ equals to zero. 
Using Steinitz' Lemma, there exists a reordering of the vectors $u_1 + \ldots + u_{(\norone{\lambda} + \norone{\gamma})}$ such that the partial sum $\sum_{i=1}^\ell u_i \leq d \Delta$, for each $\ell \leq \norone{\lambda} + \norone{\gamma}$. 

If $\norone{\lambda} + \norone{\gamma} > (2d \Delta +1)^d$ then by the pigeonhole principle, there exist two partial sums of the same value. Hence, there are two sequences that sum up to zero, i.e. there exist non-zero vectors $\lambda', \lambda'' \in \ZZgeq^{B^{(1)}}$ with $\lambda = \lambda' + \lambda''$ and $\gamma', \gamma'' \in \ZZgeq^{B^{(2)}}$ with $\gamma = \gamma' + \gamma''$ such that $B^{(1)} \lambda' - B^{(2)} \gamma'= 0$ and $B^{(1)} \lambda'' - B^{(1)} \gamma''= 0$. Hence $B^{(1)} \lambda' = B^{(2)} \gamma'$ and $B^{(1)} \lambda'' = B^{(2)} \gamma''$ are elements of the intersection cone. This implies that $b$ can be decomposed in the intersection cone.
\end{proof}
Using a similar argumentation as in the previous lemma, we can consider the intersection of several integer cones. Note that we can not simply use the above Lemma inductively as this would lead to worse bounds.
\begin{lemma}\label{lem:intersection_all}
    Consider integer cones $int.cone(B^{(1)}), \ldots , int.cone(B^{(\ell)})$ for some generating sets $B^{(1)}, \ldots , B^{(\ell)} \subset \ZZgeq^{d}$ with $\norinf{x} \leq \Delta$ for each $x \in B^{(i)}$. Consider the integer cone of the intersection 
    \begin{align*}
        int.cone(\hat{B}) = \bigcap_{i=1}^\ell int.cone(B^{(i)})
    \end{align*}
    for some generating set of elements $\hat{B}$.
    
    Then for each generating element $b \in \hat{B}$ with $B^{(i)} \lambda^{(i)} = b$ for some $\lambda^{(i)} \in \ZZgeq^{B^{(i)}}$ in the intersection cone, we have that $\norone{\lambda^{(i)}} \leq O((d \Delta)^{d (\ell-1)})$ for all $1 \leq i \leq \ell$.
\end{lemma}
\begin{proof}
Given vectors $\lambda^{(1)}, \ldots , \lambda^{(\ell)}$ with $\lambda^{(k)} \in \ZZgeq^{B^{(k)}}$ and $B^{(k)} \lambda^{(k)} = b$ for each $k \leq \ell$. Consider the sum of vectors $v^{(k)}_1 + \ldots + v^{(k)}_{\norone{\lambda^{(k)}}}$ for each $1 \leq k \leq \ell$ consisting of $\lambda_j^{(k)}$ copies of the $j$-th element of $B^{k}_j$. By adding $0$ vectors to sums we can assume without loss of generality that every sequence has the same number of summands $L = \max_{i=1, \ldots , \ell} \norone{\lambda^{(i)}}$.

\textbf{Claim:} There exists a reordering $u^{(k)}_1 + \ldots + u^{(k)}_{L}$ for each of these sums such that each partial sum $p^{(k)}_m = \sum_{i =1}^m u^{(k)}_i$ is close to the line between $0$ and $b$ and more precisely: \begin{align*}
    \norinf{p^{(k)}_m - \frac{m}{L} b} \leq 4 \Delta (d+1).
\end{align*} for each $m \leq L$ and each $k \leq \ell$.
To see this, we construct the sequence that consists of vectors from $B^{(k)}$ and subtract $L$ fractional parts of the vector $b$. To count the number of vectors we use an additional component with weight $\Delta$ of the vector and define $\bar{v}^{(k)}_i = \begin{pmatrix} \Delta \\ v^{(k)}_i \end{pmatrix}$ and $\bar{b} = \begin{pmatrix} L \Delta \\ b \end{pmatrix}$. Note that $\norone{\bar{v}^{(k)}_i},\norone{\frac{1}{L}\bar{b}} \leq 2 \Delta$. Then the sequence $\bar{v}^{(k)}_1 + \ldots + \bar{v}^{(k)}_{L} - \frac{1}{L} \bar{b} - \ldots - \frac{1}{L}\bar{b}$ sums up to zero, as $v^{(k)}_1 + \ldots + v^{(k)}_{L} = b$. Hence  we can apply the Steinitz Lemma to obtain a reordering $\bar{u}_1 + \ldots + \bar{u}_L$ for each sequence such that each partial sum $\norinf{\sum_{i=1}^m \bar{u}_i } \leq 2 \Delta (d+1)$ for each $m \leq 2 L$. Each partial sum that sums up to index $m$ contains $p$ vectors $\bar{v}^{(k)}_j$ and $q$ vectors $\frac{1}{L} b$ for some $p,q \in \ZZgeq$ with $m =p+q$. Hence $\sum^{p}_{i=1} u_i - \frac{q}{L}b \leq 2 \Delta (d+1)$. Furthermore, the $\Delta$ entry of each vector guarantees that $|p-q| \leq 2(d+1)$ which implies the statement of the claim.

Now consider the differences of a partial sum $p^{(k)}_m$ with $p^{(1)}_m$. Using the claim from above, we can now argue that $\norinf{p^{(1)}_m - p^{(k)}_m} \leq 8 \Delta (d+1)$ for each $m \leq L$ and $k \leq \ell$ as each $p^{(k)}_m$ is close to $\frac{m}{L}b$. Therefore the number of different values for $p^{(1)}_m - p^{(k)}_m$ is bounded by $(16 \Delta (d+1)+1)^d$. Assuming that $L > (16 \Delta (d+1)+1)^{d(\ell-1)}$, by the pigeonhole principle there exist indices $m'$ and $m''$ with $m' > m''$ such that $p^{(1)}_{m'} - p^{(k)}_{m'} = p^{(1)}_{m''} - p^{(k)}_{m''}$ for each $k \leq \ell$. Hence $p^{(1)}_{m'} - p^{(1)}_{m''} = \ldots = p^{(\ell)}_{m'} - p^{(\ell)}_{m''} =: b'$ and 
$b', b-b' \in \cap_{i=1}^\ell B_i$. This implies that $b$ can be decomposed and is therefore not a generating element of $\cap_{i=1}^\ell int.cone(B^i).$
\end{proof}

\subsection{Proof of Lemma \ref{lem:subrepresentation}:}
Using the results from the previous section, we are now finally able to prove our main Lemma \ref{lem:subrepresentation}.

To get an intuition for the problem however, we start by giving a sketch of the proof for the $1$-dimensional case. In this case, the multisets $T_i$ consist solely of natural numbers, i.e $T_1, \ldots , T_n \subset \ZZgeq$. Suppose that each set $T_i$ consists only of many copies of a single integral number $x_i \in \{1 , \ldots , \Delta \}$. Then it is easy to find a common multiple as $\frac{\Delta!}{1} \cdot 1 = \frac{\Delta !}{2} \cdot 2 = \ldots = \frac{\Delta !}{\Delta} \cdot \Delta$. Hence one can choose the subsets consisting of $\frac{\Delta !}{x_i}$ copies of $x_i$.
Now suppose that the multisets $T_i$ can be arbitrary. If $|T_i| \leq \Delta \cdot \Delta ! = \Delta^{O(\Delta)}$ we are done. But on the other hand, if $|T_i| \geq \Delta \cdot \Delta!$, by the pigeonhole principle there exists a single element $x_i \in \{1, \ldots , \Delta\}$ for every $T_i$ that appears at least $\Delta !$ times. Then we can argue as in the previous case where we needed at most $\Delta !$ copies of a number $x_i \in \{ 1, \ldots , \Delta \}$. This proves the lemma in the case $d=1$.

In the case of higher dimensions, the lemma seems much harder to prove. But in principle we use generalizations of the above techniques. Instead of single natural numbers however, we have to work with bases of corresponding basic feasible LP solutions and the intersection of the integer cone generated by those bases.

In the proof we need the notion of a cone which is simply the relaxation of an integer cone. For a generating set $B \subset \ZZgeq^d$, a cone is defined by
\begin{align*}
    cone(B) = \{ \sum_{b \in B} \lambda_b b \mid \lambda \in \RRgeq^B \}.
\end{align*}
\begin{proof}
First, we describe the multisets $T_1, \ldots, T_n \subset \ZZgeq^d$ by multiplicity vectors $\lambda^{(1)}, \ldots , \lambda^{(n)} \in \ZZgeq^{P}$, where $P \subset \ZZ^d$ is the set of integer points $p$ with $\norinf{p} \leq \Delta$. Each $\lambda^{(i)}_{p}$ thereby states the multiplicity of vector $p$ in $T_i$. Hence $\sum_{t \in T_i} t = \sum_{p \in P} \lambda^{(i)}_{p} p$ and our objective is to find vectors $y^{(1)}, \ldots, y^{(n)} \in \ZZgeq^P$ with $y^{(i)} \leq \lambda^{(i)}$ such that $\sum_{p \in P} y^{(1)}_{p} p = \ldots = \sum_{p \in P} y^{(n)}_{p} p$.

Consider the linear program 
\begin{align}\label{LP:basic}
     \sum_{p \in P} x_p p = b\\
     x \in \RRgeq^P \notag
     \end{align} 
Let $x^{(1)}, \ldots , x^{(\ell)} \in \RRgeq^d$ be all possible basic feasible solutions of the LP corresponding to bases $B^{(1)}, \ldots , B^{(\ell)} \in \ZZgeq^{d \times d}$ i.e. $B^{(i)} x^{(i)} = b$.

In the following we proof two claims that correspond to the two previously described cases of the one dimensional case. First, we consider the case that essentially each multiset $T_i$ corresponds to one of the basic feasible solution $x^{(j)}$. In the $1$-dimensional case this would mean that each set consists only of a single number. Note that the intersection of integer cones in dimension $1$ is just the least common multiple, i.e. $int.cone(z_1) \cap int.cone(z_2) = int.cone(lcm(z_1, z_2))$ for some $z_1, z_2 \in \ZZgeq$.

\textbf{Claim 1:} If for all $i$ we have $\norone{x^{(i)}} > d \cdot O((d \Delta)^{d (\ell-1)})$ then there exist non-zero vectors $y^{(1)}, \ldots , y^{(\ell)} \in \ZZgeq^d$ with $y^{(1)} \leq x^{(1)}, \ldots , y^{(\ell)} \leq x^{(\ell)}$ and $\norone{y^{(i)}} \leq d \cdot O((d \Delta)^{d (\ell-1)})$ such that $B^{(1)}y^{(1)} = \ldots = B^{(\ell)} y^{(\ell)}$.

\textbf{Proof of the claim:} 
Note that $B^{(i)} x^{(i)} = b$ and hence $b \in cone(B^{(i)})$. In the following, our goal is to find a non-zero point $q \in \ZZgeq^d$ such that $q = B^{(1)} y^{(1)} = \ldots = B^{(\ell)} y^{(\ell)}$ for some vectors $y^{(1)}, \ldots , y^{(\ell)} \in \ZZgeq^d$. However, this means that $q$ has to be in the integer cone $int.cone(B^{(i)})$ for every $1 \leq i \leq \ell$ and therefore in the intersection of all the integer cones, i.e. $q \in \bigcap_{i=1}^n int.cone(B^{(i)})$. By Lemma \ref{lem:intersection_all} there exists a set of generating elements $\hat{B}$ such that 
\begin{itemize}
    \item $int.cone(\hat{B}) = \bigcap_{i=1}^n int.cone(B^{(i)})$ and $int.cone(\hat{B}) \neq \{ 0 \}$ as $b \in cone(\hat{B})$ and
    \item each generating vector $p \in \hat{B}$ can be represented by $p = B^{(i)} \lambda$ for some $\lambda \in \ZZgeq^d$ with $\norone{\lambda} \leq O((d \Delta)^{d (\ell-1)})$ for each basis $B^{(i)}$.
\end{itemize}
As $b \in cone(\hat{B})$ there exists a vector $\hat{x} \in \RRgeq^{\hat{\BB}}$ with $\hat{\BB} \hat{x} = b$. Our goal is to show that there exists a non-zero vector $q \in \hat{B}$ with $\hat{x}_q \geq 1$. In this case $b$ can be simply written by $b = q + q'$ for some $q' \in cone(\hat{B})$. As $q$ and $q'$ are contained in the intersection of all cones, there exists for each generating set $\BB^{(j)}$ a vectors $y^{(j)} \in \ZZgeq^{\BB^{(j)}}$ and $z^{(j)} \in \RRgeq^{\BB^{(j)}}$ such that $\BB^{(j)} y^{(j)} = q$ and $\BB^{(j)} z^{(j)} = q'$. Hence $x^{(j)} = y^{(j)} + z^{(j)}$ and we finally obtain that $x^{(j)} \geq y^{(j)}$ for $y^{(j)} \in \ZZgeq^{\BB^{(j)}}$ which shows the claim.

Therefore it only remains to prove the existence of the point $q$ with $\hat{x}_q \geq 1$. By Lemma \ref{lem:intersection_all}, each vector $p \in \hat{B}$ can be represented, by $p = B^{(i)} x^{(p)}$ for some $x^{(p)} \in \ZZgeq^{B^{(i)}}$ with $\norone{x^{(p)}} \leq O((d \Delta)^{d (\ell-1)})$ for every basis $B^{(i)}$. 

As $B^{(i)} x^{(i)} = b = \sum_{p \in \hat{B}} \hat{x}_p p = \sum_{p \in \hat{B}} \hat{x}_p (B^{(i)} x^{(p)})$, every $x^{(i)}$ can be written by  $x^{(i)} = \sum_{p \in \hat{B}} x^{(p)} \hat{x}_p$ and we obtain a bound on $\norone{x^{(i)}}$ assuming that every for every $p \in \hat{B}$ we have $\hat{x}_p < 1$.
\begin{align*}
    \norone{x^{(i)}} \leq \sum_{p \in \hat{B}} \norone{x^{(p)} \hat{x}_p} \stackrel{\hat{x}_p < 1}{<} \sum_{p \in \hat{B}} \norone{x^{(p)}} \leq d \cdot O((d \Delta)^{d (\ell-1)}).
\end{align*}
The last inequality follows as we can assume by Caratheodory's theorem \cite{schrijver1998theory} that the number of non-zero components of $\hat{x}$ is less or equal than $d$.
Hence if $\norone{x^{(i)}} \geq d \cdot O((d \Delta)^{d (\ell-1)})$ then there has to exist a vector $q \in \hat{B}$ with $x_q \geq 1$ which proves the claim.

\textbf{Claim 2:} For every vector $\lambda^{(i)} \in \ZZgeq^P$ with $\sum_{p \in P} \lambda_p p = b$ there exists a basic feasible solution $x^{(j)}$ of LP (\ref{LP:basic}) with basis $B^{(j)}$ such that $\frac{1}{\ell}x^{(j)} \leq \lambda^{(i)}$ in the sense that $\frac{1}{\ell} x^{(j)}_p \leq \lambda^{(i)}_p$ for every $p \in B^{(j)}$.

\textbf{Proof of the claim:} 
The proof of the claim can be easily seen as each multiplicity vector $\lambda^{(i)}$ is also a solution of the linear program (\ref{LP:basic}). By standard LP theory, we know that each solution of the LP is a convex combination of the basic feasible solutions $x^{(1)}, \ldots , x^{(\ell)}$. Hence, each multiplicity vector $\lambda^{(i)}$ can be written as a convex combination of $x^{(1)}, \ldots , x^{(\ell)}$, i.e. for each $\lambda^{(i)}$, there exists a $t \in \RRgeq^\ell$ with $\norone{t} = 1$ such that $\lambda^{(i)} = \sum_{i=1}^\ell t_i \bar{x}^{(i)}$, where $\bar{x}^{(i)}_p~=~\begin{cases} x^{(i)}_p &\text{ if } p \in B^{(i)} \\ 0 & \text{ otherwise}\end{cases}$. By the pigeonhole principle, there exists for each multiplicity vector $\lambda^{(i)}$ an index $j$ with $t_j \geq  \frac{1}{\ell}$ which proves the claim.

Using the above two claims, we can now prove the claim of the lemma by showing that for each $\lambda^{(i)}$, there exist a vector $y^{(i)} \leq \lambda^{(i)}$ with bounded $1$-norm such that $\sum_{p \in P} y^{(1)}_p p = \ldots = \sum_{p \in P} y^{(n)}_p p$.

By Claim 2 we know that for each $\lambda^{(i)}$ ($1 \leq i \leq n$) we find one of the basic feasible solutions $x^{(j)}$ ($1 \leq j \leq \ell$) with $\frac{1}{\ell} x^{(j)} \leq \lambda^{(i)}$.
Applying the first claim to vectors $\frac{1}{\ell} x^{(1)}, \ldots , \frac{1}{\ell} x^{(\ell)}$ with $\frac{1}{\ell}b = \frac{1}{\ell} Bx^{(1)} = \ldots = \frac{1}{\ell} B x^{(\ell)}$, we obtain vectors $y^{(1)} \leq \frac{1}{\ell}x^{(1)}, \ldots , y^{(\ell)} \leq \frac{1}{\ell}x^{(\ell)}$ with $By^{(1)} = \ldots = B y^{(\ell)}$. Hence, we find for each $\lambda^{(i)}$ a vector $y^{(j)} \in \ZZgeq^{B^{(j)}}$ with $y^{(j)} \leq \lambda^{(i)}$.

As $\sum_{p \in P} \lambda^{(i)}_p p = b = \sum_{p \in \BB^{(j)}}x^{(j)}_p p$ and every $p \in P$ is bounded by $\norinf{p} \leq \Delta$ we know that 
\begin{align} 
    \norone{\lambda^{(i)}} \leq d \Delta \norone{x^{(j)}} \label{eq:1}
\end{align} for every $i \leq n$ and every $j \leq \ell$.
Hence if $\norone{\lambda^{(i)}} \geq d^2 \Delta \ell \cdot O((d \Delta)^{d (\ell-1)})$, we know that $\norone{\frac{1}{\ell}x^{(j)}} \geq d \cdot O((d \Delta)^{d (\ell-1)})$. Therefore, Claim 1 can be applied to  find $y^{(j)} \leq \frac{1}{\ell}x^{(j)}$ of smaller $1$-norm.

Note that $\ell$ is bounded by $\binom{|P|}{d} \leq |P|^d$ and $|P| \leq \Delta^d$ and we obtain that \begin{align*}
    \norone{y^{(j)}} \leq d^2 \Delta \ell \cdot O((d \Delta)^{d (\ell-1)}) = (d \Delta)^{O(d (\Delta^{d^2}))}.
\end{align*}
\end{proof}

\section{A Lower Bound for the Size of Graver Elements}
In this section we prove a lower bound on the size of Graver elements for a matrix where the overlapping parts contains only a single variable, i.e. $r=1$.

First, consider the matrix 
\begin{align*}
\mathcal{A} = 
\begin{pmatrix}
-1 & 2 & 0 & \cdots & 0\\
-1 & 0 & 3 & \ddots & \vdots\\
-1 & \vdots & \ddots & \ddots & 0\\
-1 & 0 & \cdots & 0 & \Delta\\
\end{pmatrix}.
\end{align*}
This matrix is of $2$-stage stochastic structure with $r=1$ and $s=1$.
We will argue that every element in $ker(\mathcal{A}) \cap (\mathbb{Z }^\Delta \setminus \{ 0 \})$ is large and therefore, the Graver elements of the matrix are large as well. We call the variable corresponding to the $i$-th column of the matrix variable $x_i$, where $x_1$ is the variable corresponding to he column with only $-1$ entries and the $x_i$ for $i>1$ is the variable corresponding to the column with entry $i$ in component $i$ and $0$ everywhere else.
Clearly, for $x \in \mathbb{Z}^{\Delta}$ to be in $ker(\mathcal{A}) \cap \ZZ^n$, we know by the first row of matrix $\mathcal{A}$ that $x_1$ has to be a multiple of $2$. By the second row of the matrix, we know that $x_1$ has to be a multiple of $3$ and so on.
Henceforth the variable $x_1$ has to be a multiple of all numbers $1, \ldots , \Delta$. Thus $x_1$ is a multiple of the least common multiple of numbers $1, \ldots , n$ which is divided by the product of all primes between $1, \ldots , n$. By known bounds for the product of all primes $\leq n$ \cite{erdos1989ramanujan}, this implies that the value of $x_1 \in 2^{\Omega(\Delta)}$, which shows that the the size of Graver elements of matrix $\mathcal{A}$ is in $(2^{\Omega(\Delta)})$.

The disadvantage in the matrix above is that the entries of the matrix are rather big. In the following we try to reduce the largest entry of the overall matrix by encoding each number $1, \ldots , \Delta$ into a submatrix.
For the encoding we use the matrix \begin{align*}
\mathcal{C} =
\begin{pmatrix}
\Delta & - 1 & 0 & \cdots & 0 \\
0 & \Delta & -1 & \ddots & \vdots\\
\vdots & \ddots & \ddots & \ddots & 0\\
0 & \cdots & 0 & \Delta & -1
\end{pmatrix},
\end{align*}
having $s$ rows and $s+1$ constraints.
Due to the first row of matrix $\mathcal{C}$, for a vector $x \in ker(\mathcal{C}) \cap \ZZ^{s+1}$ we know by the $i$-th row of the matrix that $x_i= x_{i-1} \cdot \Delta$. Hence $x_i = \Delta^{i-1} x_1$. Now we can encode in each number $z \in \{ 0, \ldots , \Delta^{s+1} -1 \}$ in an additional row by $z = \sum_{i=0}^{s} a_i(z) \Delta^i$, where $a_i(z)$ is the $i$-th number in a representation of $z$ in base $\Delta$. Hence, we consider the following matrix:
\begin{align*} \mathcal{A}' = 
    \begin{pmatrix}
    -1 & a_0(2) & \cdots & a_s (2) &  &  \\
    \cline{2-4}
    0 &\multicolumn{1}{|c}{} & \mathcal{C}&\multicolumn{1}{c|}{} & & \\
    \cline{2-4}
    -1 & & & &  a_0(3) & \cdots & a_s (3) \\
    \cline{5-7}
    0 &  &  & & \multicolumn{1}{|c}{} & \mathcal{C} & \multicolumn{1}{c|}{} & & \\
    \cline{5-7}
    \vdots & & & & & & & &  \ddots & \\
    \end{pmatrix}
\end{align*}
By the same argumentation as in matrix $\mathcal{A}$ above we know that $x_0$ has to be a multiple of each number $2, \ldots ,\Delta^{s+1}-1$. This implies that every non-zero integer vector of $ker(\mathcal{A}')$ has infinity-norm of at least $2^{\Omega(\Delta^{s})}$ is the number of rows of the block matrix. This shows the doubly exponential lower bound for the Graver complexity of $2$-stage stochastic IPs.

\section{Multi-Stage Stochastic IPs}
In this section we show that Lemma \ref{lem:subrepresentation} can also be used to get a bound on the Graver elements of matrices with a multi-stage stochastic structure. Multi-stage stochastic IPs are a well known generalization of $2$-stage stochastic IPs. For the stochastical programming background on multi-stage stochastic IPs we refer to \cite{romisch2001multistage}. Here we simply show how to solve the deterministic equivalent IP with a large constraint matrix. Regarding the augmentation framework of multi-stage stochastic IPs, it was previously known that a similar implicit bound than $2$-stage stochastic IPs also holds for multi-stage stochastic IPs. This was shown by Aschenbrenner and Hemmecke \cite{hemmecke_multi-stage} who built upon the bound of $2$-stage stochastic IPs.

In the following we define the shape of the constraint matrix $\mathcal{M}$ of a multi-stage stochastic IP. The constraint matrix consists of given block matrices $A^{(1)}, \ldots , A^{(\ell)}$ for some $\ell \in \ZZgeq$, where each block matrix $A^{(i)}$ uses a unique set of columns in $\mathcal{M}$. For a given block matrix, let $rows(A^{(i)})$ be the set of rows in $\mathcal{M}$ which are used by $A^{(i)}$.
A matrix $\mathcal{M}$ is multi-stage stochastic shape, if the following conditions are fulfilled:
\begin{itemize}
    \item There is a block matrix $A^{i_0}$ such that for every $1 \leq i \leq n$ we have $rows(A^{(i)}) \subseteq rows(A^{(i_0)})$.
    \item For two matrices $A^{(i)}, A^{(j)}$ either $rows(A^{(i)}) \subseteq rows(A^{(j)})$ or $rows(A^{(i)}) \cap rows(A^{(j)}) = \emptyset$ holds.
\end{itemize}
An example of a matrix of multi-stage stochastic structure is given in the following:
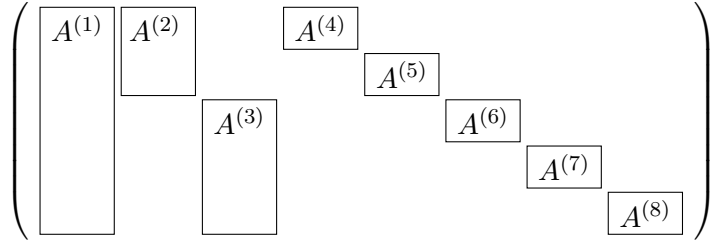
\begin{center}
\begin{tikzpicture}
\def\b{{\bullet}}
\matrix (m) [matrix of math nodes,
inner sep=3pt, column sep=3pt, row sep=2pt,
nodes={inner sep=0.25em,text width=2em,align=center},
left delimiter=(,right delimiter=),
]{%
    A^{(1)} & A^{(2)} \ & \ & A^{(4)} & \phantom{0} & \phantom{0} & \phantom{0} & \phantom{0}\\
	\phantom{0} & \phantom{0} & \phantom{0} & \ & A^{(5)} & \phantom{0} & \phantom{0} & \phantom{0} \\
	\phantom{0} & \phantom{0} & A^{(3)} & \phantom{0} & \phantom{0} & A^{(6)} & \phantom{0} & \phantom{0}\\
	\phantom{0} & \phantom{0} & \phantom{0} & \phantom{0} & \phantom{0} & \phantom{0} &A^{(7)} & \phantom{0}\\
	\phantom{0} & \phantom{0} & \phantom{0} & \phantom{0} & \phantom{0} & \phantom{0} & \phantom{0} & A^{(8)}\\
};%
\draw (m-1-1.north west) rectangle (m-5-1.south east);
\draw (m-1-2.north west) rectangle (m-2-2.south east);
\draw (m-3-3.north west) rectangle (m-5-3.south east);
\draw (m-1-4.north west) rectangle (m-1-4.south east);
\draw (m-2-5.north west) rectangle (m-2-5.south east);
\draw (m-3-6.north west) rectangle (m-3-6.south east);
\draw (m-4-7.north west) rectangle (m-4-7.south east);
\draw (m-5-8.north west) rectangle (m-5-8.south east);
\end{tikzpicture}
\end{center}

Intuitively, the constraint matrix is of multi-stage stochastic shape if the block matrices with the relation $\subseteq$ on the rows, forms a tree (see figure below).

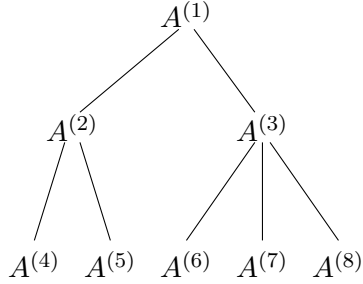
\begin{center}
\begin{tikzpicture}
	\draw (0,0) node[below]{$A^{(4)}$} -- (0.4,1.3);
	\draw (1,0) node[below]{$A^{(5)}$} -- (0.6,1.3);
	\draw (2,0) node[below]{$A^{(6)}$} -- (2.9,1.3);
	\draw (3,0) node[below]{$A^{(7)}$} -- (3,1.3);
	\draw (4,0) node[below]{$A^{(8)}$} -- (3.1,1.3);
	\draw (0.5,1.5) node{$A^{(2)}$};
	\draw (0.6,1.7) -- (1.9,2.8);
	\draw (3,1.5) node{$A^{(3)}$};
	\draw (2.9,1.7) -- (2.1,2.8);
	\draw (2,3) node{$A^{(1)}$};
\end{tikzpicture}
\end{center}
Let $s_i$ be the number of columns that are used by block matrices in the $i$-th level of the tree (starting from level $0$ at the leaves). Here we assume that the number of columns of block matrices in the same level of the tree are all identical.
Let $r$ be the number of rows that are used by the block matrices that correspond to the leaves of the tree.
In the following Lemma we show that Lemma \ref{lem:subrepresentation} can be applied inductively to bound the size of an augmenting step of multi-stage stochastic IPs. The proof is similar to that of Theorem \ref{thm:2SSIP_bound}.

\begin{theorem} \label{thm:bound_multi-stage}
    Let $y$ be an indecomposable cycle of matrix $\mathcal{M}$ then $\norinf{y}$ is bounded by a function $T(s_1, \ldots , s_{t},r, \Delta)$, where $t$ is the depth of the tree. The functions $T$ involves a tower of $t+1$ exponentials and is recursively defined by
    \begin{align*}
        & T(r, \Delta) = (\Delta r)^{O(r)}\\
        & T(s_1, \ldots , s_i,r, \Delta) = 2^{(T(s_1, \ldots , s_{i-1},r, \Delta))^{O(s^{2}_i)}}.
    \end{align*}
\end{theorem}
\begin{proof}
Consider a submatrix $\mathcal{A}$ of the constraint matirx $\mathcal{M}$ corresponding to a subtree of the tree with depth $t$. Hence, $\mathcal{A}$ itself is of multi-stage stochastic structure. Let submatrix $A \in \{A^{(1)}, \ldots , A^{(\ell)} \}$ be the root of the corresponding subtree of $\mathcal{A}$ and let the submatrices $B^{(1)}, \ldots , B^{(n)}$ be the submatrices corresponding to the subtrees of $A$ with $rows(B^{(i)}) \subseteq rows(A)$ for all $1 \leq i \leq n$.

Let $\bar{A}^{(i)}$ be the submatrix of $A$ which consists only of the rows that are used by $B^{(i)}$ (recall that $rows(B^{(i)}) \subseteq rows(A)$). Now suppose that $y$ is a cycle of $\mathcal{A}$, i.e. $\mathcal{A} y =0$ and let $y^{(0)}$ be the subvector of $y$ consisting only of the entries that belong to matrix $A$. Symmetrically let $y^{(i)}$ be the entries of vector $y$ that belong only to the matrix $B^{(i)}$ for $i>0$.  
Since $\mathcal{A}y = 0$ we also know that $\bar{A}^{(i)}y^{(0)} + B^{(i)}y^{(i)} = (\bar{A}^{(i)} B^{(i)}) \begin{pmatrix} y^{(0)} \\ y^{(i)} \end{pmatrix}= 0$ for every $1 \leq i \leq n$. 
Each vector $\begin{pmatrix} y^{(0)} \\ y^{(i)} \end{pmatrix}$ can be decomposed into a multiset of indecomposable cylces $C_i$ , i.e. 
\begin{align*} \begin{pmatrix} y^{(0)} \\ y^{(i)} \end{pmatrix}
   = \sum_{c \in C_i} c
\end{align*}
where each cycle $c \in C_i$ is a vector $c = \begin{pmatrix} c^{(0)} \\ c^{(i)} \end{pmatrix}$ consisting of subvector $c^{(0)}$ of entries that belong to matrix $A$ and a subvector $c^{(i)}$ of entries that belong to the matrix $B^{(i)}$.
Note that the matrix $(A^{(i)} B^{(i)})$ has a multi-stage stochastic structure with a corresponding tree of depth $t-1$.
Hence, by induction we can assume that each indecomposable cycle $c \in C_i$ is bounded by $\norinf{c} \leq T(s_1, \ldots , s_{t-1}, r)$ for all $1 \leq i \leq n$, where $T$ is a function that involves a tower of $t$ exponentials.
In the base case that $t=0$ and the matrix $\mathcal{A}$ only consists of one block matrix, we can bound $\norinf{c}$ by $(2 \Delta r+1)^r$ using Theorem \ref{thm:Graver_bound}.
Let $p$, be the projection that maps a cycle to the entries that belong the matrix $A$ i.e. $p(c) = p(\begin{pmatrix} c^{(0)} \\ c^{(i)} \end{pmatrix}) = c^{(0)}$.

For each vector $\begin{pmatrix} y^{(0)} \\ y^{(i)} \end{pmatrix}$ and its decomposition into cycles $C_i$ let $p(C_i) = \{ p(c) \mid c \in C_i \}$. Since
\begin{align*}
    y^{(0)} = \sum_{c \in C_1} p(c) = \ldots = \sum_{c \in C_n} p(c)
\end{align*}
we can apply Lemma \ref{lem:subrepresentation}, to obtain submultisets $S_i \subseteq p(C_i)$ of bounded size \begin{align*}
    |S_i| \leq (s_t T)^{O(s_t (T^{s_{t}^2}))}
\end{align*} with $T = T(s_1, \ldots , s_{t-1}, r, \Delta)$ such that $\sum_{x \in S_1}x = \ldots = \sum_{x \in S_n} x$. As $T(s_1, \ldots , s_{t-1}, r)$ is a function with $t$ exponentials, the bound on $|S_i|$ depends by a function $t+1$ exponentials.

There exist submultisets $\bar{C}_1 \subseteq C_1, \ldots , \bar{C}_n \subseteq C_n$ with $p(\bar{C}_1) = S_1, \ldots , p(\bar{C}_n) = S_n$. Hence, we can define the solution $\bar{y} \leq y$ by $\bar{y}^{(i)} = \sum_{c \in \bar{C}_i} \bar{p}(c)$ for every $i>0$, where $\bar{p}$ is the function that projects a vector to the entries that belong the matrix $B^{(i)}$ i.e. $\bar{p}(c) = \bar{p}(\begin{pmatrix} c^{(0)} \\ c^{(i)} \end{pmatrix}) = c^{(i)}$. For $i=0$ we define $y^{(0)} = \sum_{c \in \bar{C}_i} p(c)$. As the sum $\sum_{c \in \bar{C}_i} p(c)$ is identical for every $1 \leq i \leq n$, the vector $\bar{y}$ is a well defined.

Let $K$ be the constant derived from the $O$-notation of Lemma \ref{lem:subrepresentation} and $T = T(s_1, \ldots , s_{t-1}, r, \Delta)$, then the size of $\bar{y}$ can be bounded by
\begin{align*}
    \norinf{\bar{y}} \leq T \cdot \max_i |C_i| = T \cdot (s_t T)^{K s_t \cdot T^{(s^{2}_t)}} \leq 2^{K s_t \log(s_t T) \cdot T^{(s^{2}_t)}} \leq 2^{T^{O(s^{2}_t)}}.
\end{align*}
\end{proof}

\subsubsection*{Computing the Augmenting Step}
As a consequence for the bound of the Graver elements of the constraint matrix $\mathcal{M}$ of multi-stage stochastic IPs, we obtain by the augmentation framework an algorithm to solve multi-stage stochastic IPs. As explained in Section \ref{sec:two-stage}, the core difficulty is to compute the augmenting step $y \in Ker(\mathcal{M})$ such that $z + \lambda y$ is a feasible solution for a given initial feasible solution $z$ and a multiple $\lambda$.
Therefore, we have to solve the IP $\max \{c^T x \mid \mathcal{M}x = 0, \bar{\ell} \leq x \leq \bar{u} , \norinf{x} \leq T \}$ for some upper and lower bounds $\bar{l}, \bar{u}$ and constant $T = T(s_1, \ldots , s_t,r , \Delta)$ that is derived from the bound of Theorem \ref{thm:bound_multi-stage}.
This IP can be solved similar than in the case of $2$-stage stochastic IPs. However, since we have multiple layers, we have to apply the algorithm recursively. At each recursive call, we guess the value of the variables of the corresponding matrix and then apply the algorithm recursively. For further details on the algorithmic side and the running time we refer to \cite{hemmecke_multi-stage} or \cite{KouteckyLO18}.

As a final result we obtain the following theorem for multi-stage stochastic IPs:
\begin{theorem}
    A multi-stage stochastic IP with a constraint matrix $\mathcal{M}$ that corresponds to a tree of depth $t$ can be solved in time
    \begin{align*}
        n^2s_{0}^2 \varphi \log^2 (n s_0) \cdot T(s_1, \ldots , s_{t}, r, \Delta)
    \end{align*} where $\varphi$ is the encoding length of the IP and $T$ is a function depending only on parameters $s_1, \ldots , s_t,r, \Delta$ and involves a tower of $t+1$ exponentials.
\end{theorem}

\bibliographystyle{abbrv}
\bibliography{library}

\end{document}

%% file: fig_paths.tex
\begin{tikzpicture}[scale=0.6]
 \def\X{6}
 \def\Y{5}
 
 \draw[->] (0,0) -- (\X+1,0);
 
 \draw[->] (0,0) -- (0,\Y+1);

 \draw[->] (0,0) -- (1.4,0.2);
 \draw[->] (1.4,0.2) -- ++(1.6,0.8) node[below] {$T_1$};
 \draw[->] (3,1) -- ++(0.3,1.3);
 \draw[->] (3.3,2.3) -- ++(0.7,0.7);
 \draw[->] (4,3) -- ++(0.3,1.2);
 \draw[->] (4.3,4.2) -- ++(0.7,0.8) node[left, above] {$b$};

 \draw[->] (0,0) -- (0.2,1);
 \draw[->] (0.2,1) -- (1,1.5) node[left] {$T_3$};
 \draw[->] (1,1.5) -- ++(0.3,1.3);
 \draw[->] (1.3,2.8) -- (2.1,3.3);
 \draw[->] (2.1,3.3) -- ++(0.4,1.1);
 \draw[->] (2.5,4.4) -- ++(0.8,0.2);
 \draw[->] (3.3,4.6) -- ++(1,0.4);
 \draw[->] (4.3,5) -- (5,5);
 
 \draw[->] (0,0) -- ++(1,0.7);
 \draw[->] (1,0.7) -- ++(.6,0.7) node[right] {$T_2$};
 \draw[->] (1.6,1.4) -- ++(0.1,1);
 \draw[->] (1.7,2.4) -- ++(1,0.4);
 \draw[->] (2.7,2.8) -- ++(0.6,0.7);
 \draw[->] (3.3,3.5) -- ++(0.4,1.1);
 \draw[->] (3.7,4.6) -- ++(1.3,0.4);
\end{tikzpicture}

%% file: fig_pathsperm.tex
\begin{tikzpicture}[scale=0.6]
 \def\X{6}
 \def\Y{5}

 \draw[snake=snake,->] (-4,2) -- node[above] {permute} (-1.5,2);
 
 \draw[->] (0,0) -- (\X+1,0);
 
 \draw[->] (0,0) -- (0,\Y+1); 
 
 \draw[->] (0,0) -- (1.4,0.2);
 \draw[->] (1.4,0.2) -- ++(0.7,0.8) node[right] {$T_1$};
 \draw[->] (2.1,1) -- ++(0.3,0.8);
 \draw[->] (2.4,1.8) -- ++(0.3,1);
 \draw[->] (2.7,2.8) -- (4,3);
 \draw[->] (4,3) -- ++(0.3,1.2);
 \draw[->] (4.3,4.2) -- ++(0.7,0.8) node[left, above] {$b$};

  \draw[->] (0,0) -- ++(1,0.7);
 \draw[->] (1,0.7) -- ++(.6,0.7) node[right] {$T_2$};
 \draw[->] (1.6,1.4) -- ++(0.1,1);
 
 \draw[->] (1.7,2.4) -- ++(1,0.4) node[right] {$b'$}; 
 \draw[->] (2.7,2.8) -- ++(0.6,0.7);
 \draw[->] (3.3,3.5) -- ++(0.4,1.1);
 \draw[->] (3.7,4.6) -- ++(1.3,0.4);

 \draw[->] (0,0) -- (0.2,1);
 \draw[->] (0.2,1) -- (1,1.5) node[left] {$T_3$};
 
 \draw[->] (1,1.5) -- ++(1,0.4);
 \draw[->] (2 , 1.9) -- ++(0.7,0.8);

 \draw[->] (2.7,2.8) -- ++(0.4,1.1);
 \draw[->] (3.1,3.9) -- ++(0.8,0.2);
 \draw[->] (3.9,4.1) -- ++(0.8,0.5);
 \draw[->] (4.7,4.6) -- (5,5);

\end{tikzpicture}

%% file: main.bbl
\begin{thebibliography}{10}

\bibitem{twostage_branch_and_bound}
S.~Ahmed, M.~Tawarmalani, and N.~V. Sahinidis.
\newblock A finite branch-and-bound algorithm for two-stage stochastic integer
  programs.
\newblock {\em Mathematical Programming}, 100(2):355--377, Jun 2004.

\bibitem{hemmecke_multi-stage}
M.~Aschenbrenner and R.~Hemmecke.
\newblock Finiteness theorems in stochastic integer programming.
\newblock {\em Foundations of Computational Mathematics}, 7(2):183--227, 2007.

\bibitem{birge2011introduction}
J.~R. Birge and F.~Louveaux.
\newblock {\em Introduction to stochastic programming}.
\newblock Springer Science \& Business Media, 2011.

\bibitem{caroe1998shaped}
C.~C. Car{\o}e and J.~Tind.
\newblock L-shaped decomposition of two-stage stochastic programs with integer
  recourse.
\newblock {\em Mathematical Programming}, 83(1-3):451--464, 1998.

\bibitem{chen2018Graver}
L.~Chen, L.~Xu, and W.~Shi.
\newblock On the graver basis of block-structured integer programming.
\newblock {\em arXiv preprint arXiv:1805.03741}, 2018.

\bibitem{cook1986integer}
W.~Cook, J.~Fonlupt, and A.~Schrijver.
\newblock An integer analogue of caratheodory's theorem.
\newblock {\em Journal of Combinatorial Theory, Series B}, 40(1):63--70, 1986.

\bibitem{loera2012algebraic}
J.~A. De~Loera, R.~Hemmecke, and M.~K{\"o}ppe.
\newblock {\em Algebraic and geometric ideas in the theory of discrete
  optimization}.
\newblock SIAM, 2012.

\bibitem{dempster1981analytical}
M.~A.~H. Dempster, M.~Fisher, L.~Jansen, B.~Lageweg, J.~K. Lenstra, and
  A.~Rinnooy~Kan.
\newblock Analytical evaluation of hierarchical planning systems.
\newblock {\em Operations Research}, 29(4):707--716, 1981.

\bibitem{eisenbrand18}
F.~Eisenbrand, C.~Hunkenschr{\"{o}}der, and K.~Klein.
\newblock Faster algorithms for integer programs with block structure.
\newblock In {\em 45th International Colloquium on Automata, Languages, and
  Programming, {ICALP} 2018, July 9-13, 2018, Prague, Czech Republic}, pages
  49:1--49:13, 2018.

\bibitem{eisenbrand2018proximity}
F.~Eisenbrand and R.~Weismantel.
\newblock Proximity results and faster algorithms for integer programming using
  the {S}teinitz lemma.
\newblock In {\em Proceedings of the Twenty-Ninth Annual ACM-SIAM Symposium on
  Discrete Algorithms}, pages 808--816. SIAM, 2018.

\bibitem{erdos1989ramanujan}
P.~Erd{\"o}s.
\newblock Ramanujan and i.
\newblock In {\em Number Theory, Madras 1987}, pages 1--17. Springer, 1989.

\bibitem{koutecky2018IJCAI}
P.~Faliszewski, R.~Gonen, M.~Kouteck{\'{y}}, and N.~Talmon.
\newblock Opinion diffusion and campaigning on society graphs.
\newblock In {\em Proceedings of the Twenty-Seventh International Joint
  Conference on Artificial Intelligence, {IJCAI} 2018, July 13-19, 2018,
  Stockholm, Sweden.}, pages 219--225, 2018.

\bibitem{gade2014decomposition}
D.~Gade, S.~K{\"u}{\c{c}}{\"u}kyavuz, and S.~Sen.
\newblock Decomposition algorithms with parametric gomory cuts for two-stage
  stochastic integer programs.
\newblock {\em Mathematical Programming}, 144(1-2):39--64, 2014.

\bibitem{Graver1975foundations}
J.~E. Graver.
\newblock On the foundations of linear and integer linear programming i.
\newblock {\em Mathematical Programming}, 9(1):207--226, 1975.

\bibitem{grinberg1980value}
V.~S. Grinberg and S.~V. Sevast'yanov.
\newblock Value of the {S}teinitz constant.
\newblock {\em Functional Analysis and Its Applications}, 14(2):125--126, 1980.

\bibitem{haneveld2001optimizing}
W.~K.~K. Haneveld and M.~H. van~der Vlerk.
\newblock Optimizing electricity distribution using two-stage integer recourse
  models.
\newblock In {\em Stochastic Optimization: Algorithms and Applications}, pages
  137--154. Springer, 2001.

\bibitem{Hemmecke_two_stage_03}
R.~Hemmecke and R.~Schultz.
\newblock Decomposition of test sets in stochastic integer programming.
\newblock {\em Math. Program.}, 94(2-3):323--341, 2003.

\bibitem{jansen2018scheduling}
K.~Jansen, K.~Klein, M.~Maack, and M.~Rau.
\newblock Empowering the configuration-ip - new {PTAS} results for scheduling
  with setups times.
\newblock {\em CoRR}, abs/1801.06460, 2018.

\bibitem{jansen2018near_linear}
K.~Jansen, A.~Lassota, and L.~Rohwedder.
\newblock Near-linear time algorithm for n-fold ilps via color coding.
\newblock {\em arXiv preprint arXiv:1811.00950}, 2018.

\bibitem{kall1994stochastic}
P.~Kall and S.~W. Wallace.
\newblock {\em Stochastic programming}.
\newblock Springer, 1994.

\bibitem{kannan1987minkowski}
R.~Kannan.
\newblock Minkowski's convex body theorem and integer programming.
\newblock {\em Mathematics of operations research}, 12(3):415--440, 1987.

\bibitem{knop2017scheduling}
D.~Knop and M.~Kouteck{\'{y}}.
\newblock Scheduling meets n-fold integer programming.
\newblock {\em J. Scheduling}, 21(5):493--503, 2018.

\bibitem{knop2017combinatorial}
D.~Knop, M.~Kouteck\'y, and M.~Mnich.
\newblock {Combinatorial n-fold Integer Programming and Applications}.
\newblock In K.~Pruhs and C.~Sohler, editors, {\em 25th Annual European
  Symposium on Algorithms (ESA 2017)}, volume~87 of {\em Leibniz International
  Proceedings in Informatics (LIPIcs)}, pages 54:1--54:14, Dagstuhl, Germany,
  2017. Schloss Dagstuhl--Leibniz-Zentrum fuer Informatik.

\bibitem{KnopKM17-bribery}
D.~Knop, M.~Kouteck{\'{y}}, and M.~Mnich.
\newblock Voting and bribing in single-exponential time.
\newblock In {\em 34th Symposium on Theoretical Aspects of Computer Science,
  {STACS} 2017, March 8-11, 2017, Hannover, Germany}, pages 46:1--46:14, 2017.

\bibitem{knop2018tight}
D.~Knop, M.~Pilipczuk, and M.~Wrochna.
\newblock Tight complexity lower bounds for integer linear programming with few
  constraints.
\newblock {\em arXiv preprint arXiv:1811.01296}, 2018.

\bibitem{KouteckyLO18}
M.~Kouteck{\'{y}}, A.~Levin, and S.~Onn.
\newblock A parameterized strongly polynomial algorithm for block structured
  integer programs.
\newblock In {\em 45th International Colloquium on Automata, Languages, and
  Programming, {ICALP} 2018, July 9-13, 2018, Prague, Czech Republic}, pages
  85:1--85:14, 2018.

\bibitem{kuccukyavuz2017introduction}
S.~K{\"u}{\c{c}}{\"u}kyavuz and S.~Sen.
\newblock An introduction to two-stage stochastic mixed-integer programming.
\newblock In {\em Leading Developments from INFORMS Communities}, pages 1--27.
  INFORMS, 2017.

\bibitem{ackermanian2009}
F.~Pelupessy and A.~Weiermann.
\newblock Ackermannian lower bounds for lengths of bad sequences of monomial
  ideals over polynomial rings in two variables.
\newblock {\em Mathematical Theory and Computational Practice}, page 276, 2009.

\bibitem{romisch2001multistage}
W.~R{\"o}misch and R.~Schultz.
\newblock Multistage stochastic integer programs: An introduction.
\newblock In {\em Online optimization of large scale systems}, pages 581--600.
  Springer, 2001.

\bibitem{schrijver1998theory}
A.~Schrijver.
\newblock {\em Theory of linear and integer programming}.
\newblock John Wiley \& Sons, 1998.

\bibitem{twostage_survey}
R.~Schultz, L.~Stougie, and M.~H. Vlerk.
\newblock Two-stage stochastic integer programming: a survey.
\newblock {\em Statistica Neerlandica}, 50(3):404--416.

\bibitem{steinitz1913bedingt}
E.~Steinitz.
\newblock Bedingt konvergente reihen und konvexe systeme.
\newblock {\em Journal f{\"u}r die reine und angewandte Mathematik},
  143:128--176, 1913.

\bibitem{zhang2014finitely}
M.~Zhang and S.~K{\"u}{\c{c}}{\"u}kyavuzvuz.
\newblock Finitely convergent decomposition algorithms for two-stage stochastic
  pure integer programs.
\newblock {\em SIAM Journal on Optimization}, 24(4):1933--1951, 2014.

\end{thebibliography}
